\documentclass[letterpaper,12pt]{article}

\usepackage[margin=1in,letterpaper]{geometry} 
\usepackage[utf8]{inputenc} 

\usepackage{addon_packages}

\begin{document}
	
	\title{Expectation and Variance of the Degree of a Node in Random Spanning Trees}
    \author{\normalsize Enrique Fita Sanmartín$$,
\hspace{0.5cm} Christoph Schnörr$$, \hspace{0.5cm} Fred A. Hamprecht\\
\normalsize IWR at Heidelberg University, 69120 Heidelberg, Germany\\
\footnotesize \texttt{\{enrique.fita.sanmartin@iwr,schnoerr@math, fred.hamprecht@iwr\}}\texttt{.uni-heidelberg.de}}
	\date{}
	\maketitle
	\begin{abstract}
		We consider a Gibbs distribution over all spanning trees of an undirected, edge weighted finite graph, where, up to normalization, the probability of each tree is given by the product of its edge weights. Defining the weighted degree of a node as the sum of the weights of its incident edges, we present analytical expressions for the expectation, variance and covariance of the weighted degree of a node across the Gibbs distribution. To generalize our approach, we distinguish between two types of weight: probability weights, which regulate the distribution of spanning trees, and degree weights, which define the weighted degree of nodes. This distinction allows us to define the weighted degree of nodes independently of the probability weights. By leveraging the Matrix Tree Theorem, we show that these degree moments ultimately depend on the inverse of a submatrix of the graph Laplacian. While our focus is on undirected graphs, we demonstrate that our results can be extended to the directed setting by considering incoming directed trees instead. 
	\end{abstract}

	\section{Introduction}
	Graph theory provides a foundational framework for understanding and analyzing a wide range of complex systems, ranging from social networks to transportation networks and beyond. In graph theory, the degree of a node is defined either as the number of its neighbors or, in the case of weighted graphs, as the total sum of the weights of its incident edges. Analyzing the degree of nodes across different graph structures not only sheds light on the topological characteristics of the graph \cite{barabasiEmergenceScalingRandom1999,bollobasProbabilisticProofAsymptotic1980,chungAverageDistancesRandom2002}, but also helps to predict the behavior of dynamic processes within the network, such as propagation phenomena or synchronization \cite{laiAnalysisIdentificationMethods2022,shiozawaDataSynchronizationNode2022,zhangIdentificationKeyNodes2024}. This has implications for various practical applications, from network reliability \cite{lamSimpleMeasureNetwork2014} to the modeling of biological systems \cite{wangStatisticalIdentificationImportant2021}.

The study of expected node degree within the set of all possible spanning trees of a graph is particularly intriguing due to the unique insights it provides into the structural properties of networks under random conditions. Trees can be regarded as the most basic form of connected graphs, as they characterize the minimal properties required for connectivity. Therefore, they play a crucial role in analyzing network properties such as resilience \cite{PhysRevResearch.3.023161} or information flow \cite{Iacobelli2022EpicenterOR}. Investigating the expected degree of a node in a random spanning tree may assist in the assessment of how network structure influences random and optimal processes.

 \textbf{Related work} Previous research in this field has primarily focused on the characteristics of node degree distributions in full graphs and their random subgraphs \cite{albert2002statistical,boccaletti2006complex,newman2003structure}. However, less attention has been given to the study of expected degrees in random spanning trees. For unweighted complete graphs, it is known that the degree distribution of a node in a random spanning tree follows a binomial distribution of $1+\operatorname{Binomial}(N-2,\frac{1}{N})$.\footnote{This follows easily from the Pr\"ufer encoding of labeled trees. Specifically, a tree with $N$ nodes can be represented by a sequence of $N$-2 numbers ranging from 1 to $N$. The frequency of a number $i$ in the code corresponds to the degree of the $i$th node minus one. Since each number may appear with probability $1/N$, it follows that the node degree distribution is given by $1+\operatorname{Binomial}(N-2,\frac{1}{N})$.}
 
 Further empirical work by \cite{willemain2002distribution} utilized Monte Carlo simulations investigate the node degree distribution in maximum spanning trees of networks with randomly distributed nodes in Euclidean space. In contrast, \cite{pozrikidis2016node} proposed a method for calculating the degree distribution of a node in spanning trees of unweighted graphs, without the need to explicitly construct the trees. This method involves derivatives of the Laplacian's determinant but its efficiency diminishes as the number of neighbors of a node increases, primarily due to reliance on the inclusion-exclusion principle. More recently, \cite{rapoportCorrelationsUniformSpanning2023}  proposed an alternative approach to computing the node degree distribution using techniques from quantum mechanics, although in practice it may suffer from similar limitations. 
 
 Recent studies have examined the expected degree of a node in random spanning trees as a discrete measure of graph curvature \cite{devriendtDiscreteCurvatureGraphs2022,devriendtGraphsNonnegativeResistance2024}. In \cite{zmigrodEfficientComputationExpectations2021}, the authors introduced efficient algorithms for calculating expectations over spanning trees constrained to decomposable additive functions along the edges, with applications in natural language processing. The expected degree of a node can be considered as a specific instance of their addressed expectations. Although their algorithms also rely on the derivatives of the Laplacian's determinant, in \cref{sec:add_decomposable}, we demonstrate how these expectations can be expressed as a sum of expectations over degrees, ultimately leading to more compact formulas through our approach.
  
  Moreover, efforts have also been made to analyze properties of the degree distribution of infinite random trees \cite{lyonsProbabilityTreesNetworks2016a,moriRandomTrees2005,rudasRandomTreesGeneral2006,shinodaUniformSpanningTrees2015}.

\textbf{Contributions} In this work, we introduce analytical formulas for calculating both the expected degree of a given node, its variance, and its covariance in random spanning trees of weighted finite graphs. While previous work has focused on the unweighted degree, i.e. the number of neighbors of a node, our work also considers the weighted degree, defined as the sum of the edge weights of the incident edges of a node. We establish a Gibbs probability distribution across all spanning trees of a weighted graph, where the likelihood of each tree is proportional to the product of the edge weights within it. To generalize our approach, we distinguish between two types of weights: probability weights and degree weights. Probability weights influence the likelihood of spanning trees, whereas degree weights determine the node's weighted degree (see \cref{sec:notation-and-preliminary-results} for further details). Notably, our findings encompass the unweighted degree (number of neighbors) as a particular case, achieved when all degree weights are set to one.

Ultimately, our results are grounded in the Matrix Tree Theorem \cite{kirchhoff1847ueber,Tutte1984}, which provides an analytical formula for the total number of spanning trees in a given graph. Using this theorem we show that the derived closed-form expressions for the expected degree of a node, its variance, and its covariance can be formulated in terms of the inverse of the Laplacian of the weighted graph, after the removal of a node (see \cref{th:Expectation_Variance_Laplacian}).

Although not the primary focus of our research, a side result of our derivations provides a method for computing the complete distribution of the node's degrees in random spanning trees (see \cref{rem:complete_degree_distribution}), where the probabilities are given by the coefficients of a polynomial. It is worth noting that our expressions for the expectation, variance and covariance do not require access to the node degree distribution. Additionally, in \cref{sec:add_decomposable}, we show how our approach can be extended to compute expectations over spanning trees constrained to decomposable additive functions, as those investigated in \cite{zmigrodEfficientComputationExpectations2021}.

\textbf{Outline}
We begin with a brief overview of the notation and basic concepts in \cref{sec:notation-and-preliminary-results}. In \cref{sec:polynomial-based-computation-expectation-and-variance-of-node-degree-in-spanning-trees}, we introduce a polynomial from which we derive expressions for the expectation and variance of the degree of a given node in a random spanning tree. Specifically, \cref{th:poly_expvar} shows that the degree's expectation and variance can be expressed using the first and second derivatives of the polynomial evaluated at $-1$.

Moving on to \cref{sec:laplacian-based-computation-expectation-and-variance-of-node-degree-in-spanning-trees}, we utilize the Matrix Tree Theorem (stated in \cref{Th:Matrix_tree}) to connect the derivatives of this polynomial and the determinant of the Laplacian. This connection facilitates the derivation of analytical expressions for the expectation and variance, which are expressed in terms of the entries of the inverse of a submatrix of the Laplacian. The primary outcome of this section is detailed in \cref{th:Expectation_Variance_Laplacian}, which provides the explicit relations. Here we also present the expression for the covariance, although its proof is deferred to \cref{sec:covariance}, as it relies on a slightly different polynomial as the one introduced in \cref{sec:polynomial-based-computation-expectation-and-variance-of-node-degree-in-spanning-trees}.

In \cref{sec:relationship-between-edge-probability-and-expected-node-degree-in-spanning-trees}, we present an alternative approach for computing the expected node degree. This method, presented in \cref{th:Expected_node_degree}, links the probability of an edge incident to $v$ being present in a spanning tree with the expected degree of $v$. Although our results are framed in the undirected setting, \cref{sec:extension-to-directed-graphs} illustrates how these findings extend to directed graphs by focusing on incoming directed trees. Finally, we provide concluding remarks in \cref{sec:degree_conclusion}.
	\section{Notation}\label{sec:notation-and-preliminary-results}
	\paragraph{Edge weight functions} We consider a weighted graph $G=(V,E,w,\omega)$, where the edge weights are defined by the function $w:E\to\mathbb{R}^+$. If $w(e)=0$, it implies that $e\notin E$. In addition, we consider a second edge weight function $\omega:E\to\mathbb{R}$. We differentiate between two weight functions because, as will become clear later, the edge weights given by $w$ determine the probability distribution across the spanning trees, while those given by $\omega$ determine the weighted degree of a node. Hence, we term the weights defined by $w$ and $\omega$ as probability and degree weights, respectively. Note that contrary to the function $w$, $\omega$ may take negative values. For simplicity, we focus primarily on undirected graphs. However, \cref{sec:extension-to-directed-graphs} demonstrates the applicability of our results to directed graphs as well.

For a given node $v\in V$, let $\mathcal{N}_G(v)$ and $\mathcal{E}_G(v)$ represent the set of neighbors and edges incident to $v$ in $G$, respectively. These are defined as:
$$\mathcal{N}_G(v)\coloneqq\{u \ : \ (u,v)\in E\}, \qquad \mathcal{E}_G(v)\coloneqq\{v\}\times\mathcal{N}_G(v)=\{e\in E \ : \ v\in e \}.$$
In this context, we will use the subscript $v$ to represent the subgraph $G_v$, which only contains the edges incident to $v$. Formally, $G_v\coloneqq(V,\mathcal{E}_G(v),w|_{\mathcal{E}_G(v)})$. Conversely, $G_{\backslash v}\coloneqq(V,E\backslash\mathcal{E}_G(v)),w|_{E\backslash\mathcal{E}_G(v)})$ denotes the complementary graph of $G_v$, where the edges incident to $v$ have been removed.

Given a function $\bar{w}:E\to\mathbb{R}$ and a node $v$, we denote the $\bar{w}$-scaled version of $G$ at $v$ by $G\times_v \bar{w}$, where $\times_v$ is an operation modifying the probability weight of the edges incident to $v$ in $G$. The weight probability function $G\times_v \bar{w}$ is given by 
\begin{equation}
\label{eq:G_scaled_operation_v}
w_{G\times_v \bar{w}}(e)=\begin{cases}
w(e)\bar{w}(e) \  &v\in e\\
w(e) \ & \text{otherwise.}
\end{cases}
\end{equation}

\paragraph{Weighted degree}

The weighted degree of a node $v$ is the sum of the degree weights, $\omega(e)$, of the edges $e$ incident to $v$. Formally,
$$\degv[v][G][\omega]\coloneqq\sum_{e\in\mathcal{E}_G(v)} \omega(e)$$
Note that when $\omega(e)=1$ for all edges, the weighted degree of $v$ is equal to the number of neighbors of $v$, i.e. $\degv[v][G][\omega]=|\mathcal{N}_G(v)|$.


\paragraph{Spanning trees}
The set of all spanning trees of $G$ will be denoted by $\Setspt[G]$. If $G$ is clear from the context, we may omit the subscript.  We assign a weight to each spanning tree $\tree$, represented by the product of its probability weights: 
$$w(\tree)\coloneqq\prod_{e\in E_{\tree}}w(e).$$
Similarly, for a set of spanning trees $\mathbb{T}=\{\tree_1,\dots,\tree_m\}$, its weight is defined as the sum of the weights of the trees in the set, i.e.,
$$w(\mathbb{T})=\sum_{\tree\in\mathbb{T}}w(\tree)$$

We establish a probability distribution over the set of spanning trees where the probability of each tree $\tree$ is proportional to its weight. Mathematically, this can be expressed as:
\begin{equation}\label{eq:prob_tree_def}
	\Pr(\tree)=\frac{w(\tree)}{w\left(\Setspt[G]\right)}\propto w(\tree),
\end{equation}
where $w\left(\Setspt[G]\right)\coloneqq\displaystyle\sum_{\tree\in \Setspt[G] }w(\tree)$ is the normalization factor. This distribution can be viewed as a Gibbs distribution, where the probability of a system being in a particular state is proportional to the exponential of the negative energy of that state. In this context, the energy $E(\tree)$, of a spanning tree $\tree$, is given by the negative logarithm of the product of its tree weight: $E(\tree)=-\log(w(\tree))$ Thus, the probability of a spanning tree is proportional to $\exp(-E(\tree))$.
%
%
For unweighted graphs, i.e., when $w(e)=1$ for all $e\in E$, we obtain a uniform distribution over all spanning trees of $G$. 

For a given node $v$, let 
\[\SetsptEw[k][G,v]\coloneqq\{\tree\in\Setspt[G]\ :\ \degv[v][G][\omega]=k \}\]
represent the subset of spanning trees for which the weighted degree of $v$ is equal to $k$. 
Consequently, the expected weighted degree of $v$ in a random spanning tree is given by
\[\mathbb{E}\left[\degv[v][\tree][\omega]\right]\coloneqq\sum_{k\in \Degv[v][G][ \omega]}k\cdot\Prob\left(\tree\in \SetsptEw[k][G,v]\right)=\sum_{k\in \Degv[v][G][\omega]}k\cdot\frac{w\left(\SetsptEw[k][G,v]\right)}{w\left(\Setspt[G]\right)}.\]
Note that the variable $k$ iterates over the set of feasible weighted degrees of $v$, denoted here by $\Degv[v][G][\omega]\coloneqq\{\degv[v][\tree][\omega]\}_{\tree\in\Setspt}$.

\paragraph{The Laplacian and the Matrix Tree Theorem}
The Laplacian matrix plays an important role in the Matrix Tree Theorem (\cref{Th:Matrix_tree}), which forms the foundation of our results. The Laplacian matrix of a graph $G$ is given by
\begin{equation}
	\label{eq:Laplacian_undirected}
	L_G\coloneqq D-A,
\end{equation}
where $A\in\mathbb{R}^{|V|\times|V|}$ is the vertex-adjacency matrix of $G$, represented entry-wise as $A_{uv}=w\big((u,v)\big)$, and $D$ denotes the diagonal matrix defined as $D_{uu}=\sum_{j\in V}A_{uj}$. Note that, unless specified otherwise, the entries of the Laplacian matrix are solely determined by the weight-probability function $w$. In addition, for any node $v\in V$, $L_G^{[v]}$ will stand for the Laplacian after removing the row and column indexed by $v$. Note that the Laplacian of $G$ can be decomposed as the sum of the Laplacians of $G_{v}$ and $G_{\backslash v}$, i.e.
\begin{equation}\label{eq:decomposition_Laplacian}
L_G=L_{G_v}+L_{G_{\backslash v}}
\end{equation}

Now, we're ready to articulate the Matrix Tree Theorem, which provides a closed-form for the normalization factor of the tree probability distribution defined in equation \eqref{eq:prob_tree_def}.
\begin{theorem}[\textbf{Matrix Tree Theorem}, \cite{kirchhoff1847ueber,Tutte1984}]
	\label{Th:Matrix_tree}
	Given an edge weighted undirected graph $G=(V,E,w)$, let $\Setspt$ represent the set of all spanning trees of $G$. Then, the total weight of these spanning trees, denoted as $w(\Setspt)$, can be expressed as:
	\[w(\Setspt)=\sum_{\tree\in \Setspt }w(\tree)= \sum_{\tree\in \Setspt }\prod_{e\in E_{\tree}}w(e) =\det(L_G^{[r]}),\]
	where $r$ is an arbitrary but fixed node, and $L_G^{[r]}$ represents the Laplacian matrix after removing the row and column indexed by $r$.
\end{theorem}
In other words, the Matrix Tree Theorem states that all the cofactors of the Laplacian matrix are equal and coincide with $w(\Setspt)$. Note that when $G$ is unweighted, meaning $w(e)=1$ for all edges $e$, the theorem states that $\det(L_G^{[r]})$ equals the number of spanning trees of $G$.

	\section{Polynomial-Based Computation of Expectation and Variance of Node Degree in Spanning Trees}\label{sec:polynomial-based-computation-expectation-and-variance-of-node-degree-in-spanning-trees}
	In this section, we will define a polynomial (\cref{def:polynomials_tree_degree}), whose first and second derivatives will prove to be intricately connected to the expectation and variance of the weighted degree of a node in a spanning tree (\cref{th:poly_expvar}).

\begin{definition}\label{def:polynomials_tree_degree}
Given $G=(V,E,w,\omega)$, let $\polyw(x)$ be the polynomial defined as	
\begin{equation}\label{eq:polynomial_weights_W}
	\polyw[v](x)=\sum_{k\in\Degv[v][G]}(-1)^kw\left(\SetsptEw[k][G,v]\right)x^k,
\end{equation}
where we recall that $\SetsptEw[k][G,v]\}$ denotes the set of spanning trees in which node $v$ has weighted degree equal to $k$. Note that $\polyw[v](x)$ depends on the choice of the node $v$. However, the subscript may be omitted if there is no ambiguity regarding $v$.
\end{definition}
It is clear that $\polyw(-1)=w\left(\Setspt[G]\right),$ since the polynomial sums the probability weights of all spanning trees of $G$. Consider now the scaled version of $G$ at $v$ defined by $\polywG\coloneqq G\times_v \alpha^{\omega}$ with $\alpha^{\omega}(e)=\alpha^{\omega(e)}$ for $\alpha>0$. 

The following lemma shows that for $\alpha>0$, $\polyw(-\alpha)$ returns the total probability weight of all spanning trees of $\polywG$.
\begin{lemma}\label{lem:polynomials_count_trees}
	For $\alpha>0$, we find that $\polyw(-\alpha)=\wpolywG\left(\Setspt[\polywG]\right)$, where  
	\begin{equation}\label{eq:weight_function_powered}
	\wpolywG(e)=\begin{cases}
	\alpha^{\omega(e)}w(e) \ &\text{if } e\in\mathcal{E}_G(v),\\
	w(e)\ &\text{otherwise }
	\end{cases}.
	\end{equation}
	is the edge weight probability function of the scaled version of $G$ at $v$, denoted by $\polywG$ and defined above.
\end{lemma}
\begin{proof}
	First of all, we notice that for each $\tree\in\SetsptEw[k][G,v]$, we have
	\begin{equation}
		\begin{aligned}
			\alpha^k w(\tree)&=\alpha^k\prod_{e\in E_{\tree}}w(e)=\alpha^k\left(\prod_{e\in \mathcal{E}_{\tree}(v)}w(e)\right)\left(\prod_{\substack{e\in E_{\tree} \\e\notin\mathcal{E}_{\tree}(v)}}w(e)\right)\\
			&\overset{\star}{=}\left(\prod_{e\in \mathcal{E}_{\tree}(v)}w(e)\alpha^{\omega(e)}\right)\left(\prod_{\substack{e\in E_{\tree} \\e\notin\mathcal{E}_{\tree}(v)}}w(e)\right)\\
			&=\left(\prod_{e\in \mathcal{E}_{\tree}(v)}\wpolywG(e)\right)\left(\prod_{\substack{e\in E_{\tree} \\e\notin\mathcal{E}_{\tree}(v)}}\wpolywG(e)\right)\\&=\prod_{e\in E_{\tree}}\wpolywG(e)=\wpolywG(\tree).
		\end{aligned},
	\end{equation}
	In ($\star$) we have used the fact that $\degv[v][\tree][\omega]=\sum_{e\in \mathcal{E}_{\tree}(v)}\omega(e)=k$ since $\tree\in\SetsptEw[k][G,v]$. From here, it follows easily that $$\alpha^kw\left(\SetsptE[k][G,v]\right)=\sum_{\tree\in\SetsptEw[k][G,v]}w(\tree)\alpha^k=\sum_{\tree\in\SetsptEw[k][\polywG,v]}\wpolywG(\tree)=\wpolywG\left(\SetsptEw[k][\polywG,v]\right)$$ 
	and therefore
	\begin{equation}\label{eq:\poly(-alpha)}
	\begin{aligned}
	\polyw(-\alpha)&=\sum_{k\in\Degv[v][G]}(-1)^kw\left(\SetsptE[k][G,v]\right)(-\alpha)^k=\sum_{k\in\Degv[v][G]}w\left(\SetsptE[k][G,v]\right)\alpha^k\\
	&=\sum_{k\in\Degv[v][\polywG]}\wpolywG\left(\SetsptE[k][\polywG,v]\right)=\wpolywG\left(\Setspt[\polywG]\right)
	\end{aligned}
	\end{equation}
    \phantom{?}
\end{proof}

The next result relates the derivatives of the polynomial $\polyw(x)$ with the expectation and variance of the weighted degree of a node in a random spanning tree.
\begin{theorem}
	\label{th:poly_expvar}
	Let $G=(V,E,w,\omega)$ be a graph and $v$ be an arbitrary but fixed node. Then the expectation and variance of the weighted degrees of node $v$ in $G$ are given by:
	\begin{equation}
	\label{eq:poly_exp} \mathbb{E}\left[\degv[v][\tree][\omega]\right]=-\frac{\polyw'(-1)}{\polyw(-1)}
	\end{equation}
	\begin{equation}
	\label{eq:polyw_var}
	\Var\left[\degv[v][\tree][\omega]\right]=\frac{\partial\left(-x\frac{\polyw'(-x))}{\polyw(-x)}\right)}{\partial x}\Biggr|_{x=1}=\frac{\polyw''(-1)}{\polyw(-1)}-\frac{\polyw'(-1)\polyw'(-1)}{\polyw^2(-1)}-\frac{ \polyw'(-1)}{\polyw(-1)}.
	\end{equation}
	where $\polyw(x)$ is defined as in \cref{def:polynomials_tree_degree}.
\end{theorem}
\begin{proof}
	The derivative of $\polyw(x)$ is given by:
	\begin{equation}
	\label{eq:th_poly_expvar:diff_\poly(x)}
	\polyw'(x)=\frac{\partial \polyw(x)}{\partial x}=\sum_{k=1}^{|V|}(-1)^kkw\left(\SetsptE[k][G,v]\right)x^{k-1}
	\end{equation}
	Therefore, combining \eqref{eq:\poly(-alpha)} and \eqref{eq:th_poly_expvar:diff_\poly(x)}, we easily derive \eqref{eq:poly_exp}:
	\begin{equation}
	\label{eq:th_poly_expvar:Expect=diffp}
	\begin{aligned}
	-\frac{\polyw'(-1)}{\polyw(-1)}&=-\frac{\sum_{k\in\Degv[v][G][\omega]}(-1)^kkw\left(\SetsptE[k][G,v]\right)(-1)^{k-1}}{w\left(\Setspt[G]\right)}=\sum_{k\in\Degv[v][G][\omega]}k\frac{w\left(\SetsptE[k][G,v]\right)}{w\left(\Setspt[G]\right)}\\
	&=\sum_{k\in\Degv[v][G][\omega]}k\Prob\left(\tree\in\SetsptE[k][G]\right)=\mathbb{E}\left[\degv[v][\tree][\omega]\right]
	\end{aligned}
	\end{equation}

	Now we will prove the expression for the variance stated in \eqref{eq:polyw_var}. Based on equation \eqref{eq:th_poly_expvar:diff_\poly(x)}, we deduce that:
	\begin{equation}
	\label{eq:1diff_x\poly'(x)}
	\begin{aligned}
	x\frac{\partial x\polyw'(x)}{\partial x}&=x\frac{\partial \sum_{k\in\Degv[v][G][\omega]}(-1)^{k}kw\left(\SetsptE[k][G,v]\right)x^{k}}{\partial x}\\&=	x\sum_{k\in\Degv[v][G][\omega]}(-1)^kk^2w\left(\SetsptE[k][G,v]\right)x^{k-1}\\&=\sum_{k\in\Degv[v][G][\omega]}(-1)^kk^2w\left(\SetsptE[k][G,v]\right)x^{k}
	\end{aligned}
	\end{equation}
	On the other hand, using the chain rule, we obtain:
	\begin{equation}
	\label{eq:2diff_xp'(x)}
	x\frac{\partial \left[x\polyw'(x)\right]}{\partial x}=x^2\polyw''(x)+x\polyw'(x)
	\end{equation}
	Therefore, following the same reasoning as in equation \eqref{eq:th_poly_expvar:Expect=diffp}, we derive that
	\begin{equation}
	\begin{aligned}
	\frac{ \polyw''(-1)-\polyw'(-1)}{\polyw(-1)}&=\frac{\sum_{k\in\Degv[v][G][\omega]}k^2w\left(\SetsptE[k][G,v]\right)(-1)^{2k}}{\polyw(-1)}\\
	&=\frac{\sum_{k\in\Degv[v][G][\omega]}k^2w\left(\SetsptE[k][G,v]\right)}{w\left(\Setspt[G]\right)}\\
	&=\sum_{k\in\Degv[v][G][\omega]}k^2\Prob\left(t\in\SetsptE[k][G]\right)=\mathbb{E}\left[(\degv[v][\tree][\omega])^2\right]
	\end{aligned}
	\end{equation}
	Consequently,
	\begin{equation}
	\label{eq:Var=p}
	\Var\left[\degv[v][\tree][\omega]\right]=\mathbb{E}\left[(\degv[v][\tree][\omega])^2\right]-\left(\mathbb{E}\left[\degv[v][\tree][\omega]\right]\right)^2=\frac{\polyw''(-1)}{\polyw(-1)}-\frac{\polyw'(-1)}{\polyw(-1)}-\frac{\polyw'(-1)\polyw'(-1)}{\polyw^2(-1)}.
	\end{equation}
	Finally, notice that		
	\begin{equation}
	\begin{aligned}
	\frac{\partial\left(-x\frac{\polyw'(-x))}{\polyw(-x)}\right)}{\partial x}\Biggr|_{x=1}&=\left[\frac{x \polyw''(-x)}{\polyw(-x)}-\frac{ \polyw'(-x)}{\polyw(-x)}-\frac{x \polyw'(-x)\polyw'(-x)}{\polyw^2(-x)}\right]\Biggr|_{x=1}\\
	&=\frac{\polyw''(-1)}{\polyw(-1)}-\frac{\polyw'(-1)}{\polyw(-1)}-\frac{\polyw'(-1)\polyw'(-1)}{\polyw^2(-1)}\\
	&=\Var\left[\degv[v][\tree][\omega]\right]
	\end{aligned}
	\end{equation}
    \phantom{?}
\end{proof}

\begin{remark}
	A similar approach can be applied to obtain the covariance between $\degv[v]$ and $\degv[u]$ for two nodes $u$ and $v$. However, this involves generalizing the polynomial $\polyw(x)$ to depend on two variables. More details on this generalization are provided in \cref{sec:covariance}.
\end{remark}

	\section{Laplacian-Based Computation of Expectation and Variance of Node Degree in Spanning Trees}\label{sec:laplacian-based-computation-expectation-and-variance-of-node-degree-in-spanning-trees}
    \nopagebreak
	
\def\backgroundcolortext{cyan }
\def\backgroundcolor{cyan!10}
\def\semiringcolor{red}
\def\semiringcolortext{red }


\begin{center}
\begin{table}[t]
	\centering 

		\scalebox{0.82}{
			\begin{tabular}{ccc}

                \multicolumn{2}{c}{\textbf{Graph}}&\textbf{Description}\\
                \hline
                
                \raisebox{-0.5\height}{\begin{tikzpicture}
\tikzset{mynode/.style={draw, circle, inner sep=1pt, font=\scriptsize}}
\tikzset{myedge/.style={blue, line width=1pt}}
\def\scale_g{0.5}

\node[mynode] (1) at (\scale_g,0) {$u$};
\node[mynode] (2) at (-\scale_g,0) {\phantom{2}};
\node[mynode] (3) at (0,\scale_g) {\phantom{3}};
\node[mynode] (4) at (0,-\scale_g) {\phantom{4}};
\node[mynode] (5) at (-\scale_g*2,0) {\phantom{5}};
\node[mynode] (v) at (0,0) {$v$};

\node at (\scale_g*1.7,\scale_g*1.0) {\phantom{\tiny$\times \omega(e)^p$}};


\draw (v) -- (1) node[midway,  above] {};
\draw (v) -- (2) node[midway,  above] {};
\draw (v) -- (3) node[midway,  above, left, pos=0.4] {};
\draw (v) -- (4) node[midway,  above, right, pos=0.4] {};


\draw (1) -- (3) node[midway, sloped, above] {};
\draw (3) -- (2) node[midway, sloped, above] {};
\draw (2) -- (4) node[midway, sloped, below] {};
\draw (2) -- (5) node[midway, sloped, below] {};
\draw (4) -- (1) node[midway, sloped, below] {};

\end{tikzpicture}} &$G$&\scriptsize Base graph\\\hline

                \raisebox{-0.5\height}{\begin{tikzpicture}
\tikzset{mynode/.style={draw, circle, inner sep=1pt, font=\scriptsize}}
\tikzset{myedge/.style={blue, line width=1.5pt}}
\def\scale_g{0.5}

\node[mynode] (1) at (\scale_g,0) {$u$};
\node[mynode] (2) at (-\scale_g,0) {\phantom{2}};
\node[mynode] (3) at (0,\scale_g) {\phantom{3}};
\node[mynode] (4) at (0,-\scale_g) {\phantom{4}};
\node[mynode] (5) at (-\scale_g*2,0) {\phantom{5}};
\node[mynode] (v) at (0,0) {$v$};

\node at (\scale_g*1.7,\scale_g*1.0) {\phantom{\tiny$\times \omega(e)^p$}};


\draw[] (v) -- (1) node[midway,  above] {};
\draw[] (v) -- (2) node[midway,  above] {};
\draw[] (v) -- (3) node[midway,  above, left, pos=0.4] {};
\draw[] (v) -- (4) node[midway,  above, right, pos=0.4] {};


\end{tikzpicture}} &$G_v$&\scriptsize Graph neighborhood of $v$\\\hline
                
                \raisebox{-0.5\height}{\begin{tikzpicture}

\tikzset{mynode/.style={draw, circle, inner sep=1pt, font=\scriptsize}}
\tikzset{myedge/.style={blue, line width=1.5pt}}
\def\scale_g{0.5}

\node[mynode] (1) at (\scale_g,0) {$u$};
\node[mynode] (2) at (-\scale_g,0) {\phantom{2}};
\node[mynode] (3) at (0,\scale_g) {\phantom{3}};
\node[mynode] (4) at (0,-\scale_g) {\phantom{4}};
\node[mynode] (5) at (-\scale_g*2,0) {\phantom{5}};
\node[mynode] (v) at (0,0) {$v$};

\node at (\scale_g*1.7,\scale_g*1.0) {\phantom{\tiny$\times \omega(e)^p$}};




\draw (1) -- (3) node[midway, sloped, above] {};
\draw (3) -- (2) node[midway, sloped, above] {};
\draw (2) -- (4) node[midway, sloped, below] {};
\draw (2) -- (5) node[midway, sloped, below] {};
\draw (4) -- (1) node[midway, sloped, below] {};
\end{tikzpicture}} &$G_{\backslash v}$&\makecell{\scriptsize Complementary graph of $G_v$}\\\hline
                
                \raisebox{-0.5\height}{\begin{tikzpicture}
\tikzset{mynode/.style={draw, circle, inner sep=1pt, font=\scriptsize}}
\tikzset{myedge/.style={black, line width=1.5pt}}
\def\scale_g{0.5}

\node[mynode] (1) at (\scale_g,0) {$u$};
\node[mynode] (2) at (-\scale_g,0) {\phantom{2}};
\node[mynode] (3) at (0,\scale_g) {\phantom{3}};
\node[mynode] (4) at (0,-\scale_g) {\phantom{4}};
\node[mynode] (5) at (-\scale_g*2,0) {\phantom{5}};
\node[mynode] (v) at (0,0) {$v$};

\node at (\scale_g*1.7,\scale_g*1.0) {\tiny$\times \omega(e)^p$};


\draw[myedge] (v) -- (1) node[midway,  above] {};
\draw[myedge] (v) -- (2) node[midway,  above] {};
\draw[myedge] (v) -- (3) node[midway,  above, left, pos=0.4] {};
\draw[myedge] (v) -- (4) node[midway,  above, right, pos=0.4] {};


\end{tikzpicture}} &$\Gwp[p]_v\coloneqq G_v\times_v\omega^p$& \makecell{\scriptsize Probality weight of the edge, $e$, incident to $v$ is scaled by $\omega(e)^p$.\\\scriptsize Rest of probality weights are set to $0$}\\\hline
                
                \raisebox{-0.5\height}{\begin{tikzpicture}
\tikzset{mynode/.style={draw, circle, inner sep=1pt, font=\scriptsize}}
\tikzset{myedge/.style={black, line width=1.5pt}}
\def\scale_g{0.5}

\node[mynode] (1) at (\scale_g,0) {$u$};
\node[mynode] (2) at (-\scale_g,0) {\phantom{2}};
\node[mynode] (3) at (0,\scale_g) {\phantom{3}};
\node[mynode] (4) at (0,-\scale_g) {\phantom{4}};
\node[mynode] (5) at (-\scale_g*2,0) {\phantom{5}};
\node[mynode] (v) at (0,0) {$v$};

\node at (\scale_g*1.7,\scale_g*1.0) {\tiny$\times \omega(e)^p$};


\draw[myedge] (v) -- (1) node[midway,  above] {};


\end{tikzpicture}} &$\left(\Gwp[p]_{v}\right)_{u}\coloneqq \left(G_v\times_v \omega^p\right)_u$& \makecell{\scriptsize Probality weight of the edge, $(u,v)$ is scaled by $\omega(e)^p$.\\\scriptsize Rest of probality weights are set to $0$.\\\scriptsize If $(u,v)\notin E_G$, then all probability weights are 0.}\\\hline                
                
                \raisebox{-0.5\height}{\begin{tikzpicture}

\tikzset{mynode/.style={draw, circle, inner sep=1pt, font=\scriptsize}}
\tikzset{myedge/.style={black, line width=1.5pt}}
\def\scale_g{0.5}

\node[mynode] (1) at (\scale_g,0) {$u$};
\node[mynode] (2) at (-\scale_g,0) {\phantom{2}};
\node[mynode] (3) at (0,\scale_g) {\phantom{3}};
\node[mynode] (4) at (0,-\scale_g) {\phantom{4}};
\node[mynode] (5) at (-\scale_g*2,0) {\phantom{5}};
\node[mynode] (v) at (0,0) {$v$};

\node at (\scale_g*1.7,\scale_g*1.0) {\tiny$\times \alpha^{\omega(e)}$};


\draw[myedge] (v) -- (1) node[midway,  above] {};
\draw[myedge] (v) -- (2) node[midway,  above] {};
\draw[myedge] (v) -- (3) node[midway,  above, left, pos=0.4] {};
\draw[myedge] (v) -- (4) node[midway,  above, right, pos=0.4] {};


\draw (1) -- (3) node[midway, sloped, above] {};
\draw (3) -- (2) node[midway, sloped, above] {};
\draw (2) -- (4) node[midway, sloped, below] {};
\draw (2) -- (5) node[midway, sloped, below] {};
\draw (4) -- (1) node[midway, sloped, below] {};
\end{tikzpicture}} &$\polywG\coloneqq G\times_v\alpha^{\omega}$&\makecell{\scriptsize Probality weight of the edge, $e$, incident to $v$ is scaled by $\alpha^{\omega(e)}$.\\\scriptsize Rest of probality weights are not modified.}\\\hline
                
                \raisebox{-0.5\height}{\begin{tikzpicture}
\tikzset{mynode/.style={draw, circle, inner sep=1pt, font=\scriptsize}}
\tikzset{myedge/.style={black, line width=1.5pt}}
\def\scale_g{0.5}

\node[mynode] (1) at (\scale_g,0) {$u$};
\node[mynode] (2) at (-\scale_g,0) {\phantom{2}};
\node[mynode] (3) at (0,\scale_g) {\phantom{3}};
\node[mynode] (4) at (0,-\scale_g) {\phantom{4}};
\node[mynode] (5) at (-\scale_g*2,0) {\phantom{5}};
\node[mynode] (v) at (0,0) {$v$};

\node at (\scale_g*1.7,\scale_g*1.0) {\tiny$\times \alpha^{\omega(e)}$};


\draw[myedge] (v) -- (1) node[midway,  above] {};
\draw[myedge] (v) -- (2) node[midway,  above] {};
\draw[myedge] (v) -- (3) node[midway,  above, left, pos=0.4] {};
\draw[myedge] (v) -- (4) node[midway,  above, right, pos=0.4] {};


\end{tikzpicture}} &$\polywG_{v}\coloneqq G_v\times_v\alpha^{\omega}$&\makecell{\scriptsize Probality weight of the edge, $e$, incident to $v$ is scaled by $\alpha^{\omega(e)}$.\\\scriptsize Rest of probality weights are set to $0$}\\\hline

		\end{tabular}
	}
	
	\caption[]{Notation of graphs defined throughout the manuscript. The thick edges indicate those that are scaled by the corresponding term.}
	\label{tab:graph_table}

\end{table}%
\end{center}

In this section, we first show that the polynomial $\polyw(x)$ introduced in \cref{def:polynomials_tree_degree} can be obtained from the determinant of the Laplacian matrix of $\polywG$, the scaled version of $G$ at $v$ used in \cref{lem:polynomials_count_trees}. Building on the insights from \cref{th:poly_expvar} and the Matrix Tree Theorem (\cref{Th:Matrix_tree}), we establish in \cref{th:Expectation_Variance_Laplacian} a connection between the derivatives of this polynomial and the determinant of the Laplacian of $\polywG$. This connection allows us to derive analytical expressions for the expected weighted degree and its variance, expressed in terms of the inverse of a submatrix of the Laplacian. We will also state an expression for the covariance, though its proof is deferred to \cref{sec:covariance}.

Given an arbitrary node $r$ of $G$, we deduce from the Matrix Tree Theorem (\cref{Th:Matrix_tree}) and \cref{lem:polynomials_count_trees} that 
\begin{align}\label{eq:poly_Laplacian_relation1}
\det \left(L_{\polywG}^{[r]}\right)=\wpolywG\left(\Setspt[\polywG]\right)=\polyw(-\alpha),
\end{align}
where the exponent $[r]$ indicates that the row and column indexed by $r$ have been removed from the matrix. Thus, the polynomial $\polyw(x)$ can be expressed as the determinant of the Laplacian of $\polywG$.  \Cref{th:Expectation_Variance_Laplacian} leverages the relationship established in \eqref{eq:poly_Laplacian_relation1} to obtain expressions for the expectation and variance of the weighted degree of a node $v$ in terms of the Laplacian matrix. Specifically, \cref{th:Expectation_Variance_Laplacian} applies \cref{th:poly_expvar}, which relates the derivatives of $\polyw$ to the expectation and variance of the degree of $v$ in a random spanning tree. 

To prepare for the proof of \cref{th:Expectation_Variance_Laplacian}, we first introduce a new scaled version of $G$ at $v$. For $p\in\mathbb{N}$, let $\Gwp\coloneqq G\times_v \omega^p$, with $\omega^p(e)=(\omega(e))^p$. As will be shown next, the Laplacian of $\Gwp$ for $p=1,2$ will be useful in expressing the expectation and variance of the degree of a node in a random spanning tree. In concrete, we will show that its Laplacian is related to the derivative of the Laplacian of $\polywG$ as a function of $\alpha$.

\begin{theorem}
	\label{th:Expectation_Variance_Laplacian}
	Let $G=(V,E,w,\omega)$ be an undirected, probability-weighted connected graph and $r$ an arbitrary node of $G$. For a given node $v$, we have that:
	\begin{align}
	\mathbb{E}\left[\degv[v][\tree][\omega]\right]&=\operatorname{Tr}\left[\left(\Lv[1]\right)^{[r]}\left(L_G^{[r]}\right)^{-1}\right],\\
	\Var\left[\degv[v][\tree][\omega]\right]&=\operatorname{Tr}\left[\left(\left(\Lv[2]\right)^{[r]}-\left(\Lv[1]\right)^{[r]} \left(L_G^{[r]}\right)^{-1}\left(\Lv[1]\right)^{[r]} \right)\left(L_G^{[r]}\right)^{-1}\right]\label{eq:variance_laplacian_main},
	\end{align}%
	where $L_G$ is the Laplacian of $G$, $\Lv[p]$ is the Laplacian matrix of $\Gwp\coloneqq G\times_v \omega^p$ as defined above, $L_G^{[r]}$ and $\left(\Lv[p]\right)^{[r]}$ represent the corresponding matrices once the row and column indexed by node $r$ have been removed and $\operatorname{Tr}$ is the trace operator.
	
	Moreover, for a second node $u$, the covariance between the $\degv[v][\tree][\omega]$ and $\degv[u][\tree][\omega]$ is given by
	\begin{equation}\label{eq:covariance_laplacian_main}
	\begin{aligned}
	\Cov(\degv[v][\tree][\omega],\degv[u][\tree][\omega])=\operatorname{Tr}\left[\left(\left(\Lvu[2][v][u]\right)^{[r]}-\left(\Lv[1]\right)^{[r]} \left(L_G^{[r]}\right)^{-1}\left(\Lu[1]\right)^{[r]} \right)\left(L_G^{[r]}\right)^{-1}\right]
	\end{aligned}
	\end{equation}
\end{theorem}
\begin{remark}
	In the proof of \cref{th:Expectation_Variance_Laplacian} we will utilize the following derivative rules for matrices \cite{magnusMatrixDifferentialCalculus2019}:
	\begin{equation}\label{eq:matrix_derivatives}
	\begin{aligned}
	&\partial\left(\log\left(\det (X)\right)\right)=\operatorname{Tr}\left[(\partial X) X^{-1}\right], \qquad &\partial\left(\operatorname{Tr}(X)\right)=\operatorname{Tr}(\partial X), \\ &\partial\left(XY\right)=\partial(X)Y+X\partial(Y), \qquad &\partial\left(X^{-1}\right)=-X^{-1}\left(\partial X\right)X^{-1}.
	\end{aligned}
	\end{equation}
\end{remark}
\begin{proof}
	We will only prove the expressions for the expectation and variance, as the covariance expression is derived from a more general polynomial that generalizes the one introduced in \cref{def:polynomials_tree_degree}. Since the reasoning for deriving the covariance is analogous to that for the expectation and variance once the appropriate polynomial is defined, we defer its proof to \cref{sec:covariance}.
	
	First, notice the following relation:
	\begin{equation}\label{eq:relation log_poly with expectation}
	x\frac{\partial \log \poly(-x)}{\partial x}=-x\frac{\poly'(-x)}{\poly(-x)}\xrightarrow[\text{\cref{th:poly_expvar}}]{x=1} \mathbb{E}\left[\degv[v][\tree]\right],
	\end{equation}
	where the last part is a consequence of equation \eqref{eq:poly_exp} stated in \cref{th:poly_expvar}. Therefore, in combination with equation \eqref{eq:poly_Laplacian_relation1}, we deduce that we just need to derive $\log\det\left(L_{\polywG}^{[r]}\right)$ at $\alpha=1$. From the derivative rules stated in \eqref{eq:matrix_derivatives} we know that 	
	\begin{equation}\label{eq:derivative_log_det_L}
	\alpha\frac{\partial \log\det\left(L_{\polywG}^{[r]}\right)}{\partial \alpha}=\alpha\operatorname{Tr}\left[\frac{\partial  L_{\polyG}^{[r]}}{\partial \alpha}\left(L_{\polywG}^{[r]}\right)^{-1}\right]
	\end{equation}
	Similarly as in equation \eqref{eq:decomposition_Laplacian}, the Laplacian of $\polywG$ can be decomposed as
	\begin{equation}\label{eq:Laplacians_poly_graphs}
	L_{\polywG}=L_{\polywG_{\backslash v}}+L_{\polywG_{v}}=L_{G_{\backslash v}}+L_{\polywG_{v}}.
	\end{equation}
	Since the probability weights of $G$ and $\polywG$ only differ on the edges incident to $v$, only the second term of the sum depends on $\alpha$. Thus, we have that the derivative evaluated at $\alpha=1$ is given by
	\begin{equation}\label{eq:derivative_Laplacian}
	\begin{aligned}
	\frac{\partial L_{\polywG}}{\partial \alpha}\Biggr|_{\alpha=1}&=\frac{\partial L_{\polywG_{v}}}{\partial \alpha}\Biggr|_{\alpha=1}=\frac{\partial L_{G_v\times_v \alpha^{\omega}}}{\partial \alpha}\Biggr|_{\alpha=1}=L_{G_v\times_v\omega\alpha^{\omega-1}}\Biggr|_{\alpha=1}= L_{G_v\times_v\omega}\eqqcolon\Lv[1].
	\end{aligned}
	\end{equation}
	Finally, combining equations \eqref{eq:poly_Laplacian_relation1}, \eqref{eq:relation log_poly with expectation}, \eqref{eq:derivative_log_det_L} and \eqref{eq:derivative_Laplacian} we obtain that%
	\begin{equation}\label{eq1:Expectation_Variance_Laplacian}
	\begin{aligned}
	\mathbb{E}\left[\degv[v][\tree][\omega]\right]&\underbrace{=}_{\eqref{eq:relation log_poly with expectation}}\alpha\frac{\partial \log \polyw(-\alpha)}{\partial \alpha}\Biggr|_{\alpha=1}\underbrace{=}_{\eqref{eq:poly_Laplacian_relation1}}\alpha\frac{\partial \log\det\left(L_{\polywG}^{[r]}\right)}{\partial \alpha}\Biggr|_{\alpha=1}\\
	&\underbrace{=}_{\eqref{eq:derivative_log_det_L}}\operatorname{Tr}\left[\alpha\frac{\partial L_{\polyG}^{[r]}}{\partial \alpha}\left(L_{\polywG}^{[r]}\right)^{-1}\right]\Biggr|_{\alpha=1}\underbrace{=}_{\eqref{eq:derivative_Laplacian}}\operatorname{Tr}\left[\left(\Lv[1]\right)^{[r]}\left(L_G^{[r]}\right)^{-1}\right].
	\end{aligned}
	\end{equation}

	Similarly, we will prove the variance equality. We know from equation \eqref{eq:th_poly_expvar:Expect=diffp} of \cref{th:poly_expvar} that $$\Var\left[\degv[v][\tree]\right]=\frac{\partial\left(-x\frac{\poly'(-x))}{\poly(-x)}\right)}{\partial x}\Biggr|_{x=1}.$$
	Thus, we just need to derive \eqref{eq1:Expectation_Variance_Laplacian} with respect to $\alpha$ and evaluate it at $1$:
	\begin{equation*}\label{eq3:Expectation_Variance_Laplacian}
	\begin{aligned}
	\Var\left[\degv[v][\tree][\omega]\right]&=\frac{\partial\operatorname{Tr}\left[\alpha\frac{\partial L_{\polyG}^{[r]}}{\partial \alpha}\left(L_{\polywG}^{[r]}\right)^{-1}\right]}{\partial \alpha}\Biggr|_{\alpha=1}\\&=\frac{\partial\operatorname{Tr}\left[\alpha \left(L_{G_v\times_v\omega\alpha^{\omega-1}}\right)^{[r]}\left(L_{\polywG}^{[r]}\right)^{-1}\right]}{\partial \alpha}\Biggr|_{\alpha=1}	\\
	&\underbrace{=}_{\eqref{eq:derivative_Laplacian}\& \eqref{eq:matrix_derivatives} }\operatorname{Tr}\left[\left(L_{G_v\times_v\omega^2\alpha^{\omega-1}}\right)^{[r]}\left(L_{\polyG}^{[r]}\right)^{-1}\right]\Biggr|_{\alpha=1}\\
	&-\operatorname{Tr}\left[\left(L_{G_v\times_v\omega\alpha^{\omega}}\right)^{[r]}\left(L_{\polyG}^{[r]}\right)^{-1}\left(L_{G_v\times_v\omega\alpha^{\omega}}\right)^{[r]}\left(L_{\polyG}^{[r]}\right)^{-1}\right]\Biggr|_{\alpha=1}\\
	&=\operatorname{Tr}\left[\left(\Lv[2]\right)^{[r]}\left(L_G^{[r]}\right)^{-1}-\left(\Lv[1]\right)^{[r]} \left(L_G^{[r]}\right)^{-1}\left(\Lv[1]\right)^{[r]} \left(L_G^{[r]}\right)^{-1}\right]\\
	&=\operatorname{Tr}\left[\left(\left(\Lv[2]\right)^{[r]}-\left(\Lv[1]\right)^{[r]} \left(L_G^{[r]}\right)^{-1}\left(\Lv[1]\right)^{[r]} \right)\left(L_G^{[r]}\right)^{-1}\right].
	\end{aligned}
	\end{equation*}
	\vphantom{.}
\end{proof}

\begin{remark}[\textbf{Computation probability distribution of node degree}]\label{rem:complete_degree_distribution}
	It should be noted that while our focus has been on computing the expectation and variance of the degree of a node in a spanning tree, equation \eqref{eq:poly_Laplacian_relation1} offers a method for calculating the entire probability distribution of the node degree. By evaluating the determinant with respect to $\alpha$ and dividing by the total weight of the trees, the absolute values of the coefficients multiplying the monomials $\alpha^i$ reveal the probability of sampling a tree where the degree of node $v$ is equal to $i$. This method appears to offer greater efficiency compared to the approach outlined in \cite{pozrikidis2016node}, which also relies on determinant computation but additionally employs the inclusion-exclusion principle on edges incident to $v$. Consequently, it may not scale optimally with the number of neighbors of $v$. Furthermore, our method extends to weighted graphs, whereas the method described in the aforementioned reference is limited to unweighted ones.

 Moreover, the result presented in \cref{th:Expectation_Variance_Laplacian} can also be applied to compute other expectations over spanning trees when constrained to additively decomposable functions along the edges \cite{zmigrodEfficientComputationExpectations2021}. We refer to \cref{sec:add_decomposable} for further details.
\end{remark}

\begin{example}[\textbf{Node $v$ incident to all nodes with constant degree and probability weight functions at \bm{$\mathcal{E}_v$}}]\label{example:all_adjacent_exp}
	Consider the scenario where node $v$ is connected to all other nodes, i.e. $\mathcal{N}_v=V\backslash\{v\}$. Moreover, assume that the edge weights of all edges incident to $v$ are equal, meaning that $w(e)=\kappa_1$ and $\omega(e)=\kappa_2$, for all $e\in\mathcal{E}_v$. Let $\bar{G}=(V\backslash\{v\},E\backslash\mathcal{E}_G(v))$ represent the graph obtained by removing node $v$ from $G$. In this case, the Laplacian matrix $\Lv[1]^{[v]}$ is $\kappa_1\kappa_2$ times the identity matrix. Moreover, we have that
	$$L^{[v]}_{G}=L_{\bar{G}}+\kappa_1I^{[v]}.$$
	
	Let $\{\mu_i\}$ and $\{\lambda_i\}$ be the eigenvalues of $L^{[v]}_{G}$ and $L_{\bar{G}}$, respectively. They are related as follows:
	\[\mu_i=\lambda_i+\kappa_1.\]

	Thus, from \cref{th:Expectation_Variance_Laplacian}, it follows that when a node $v$ is connected to all other nodes with equal edge weights, the expected weighted degree of $v$ is related to the eigenvalues of the Laplacian of the graph with $v$ removed:		
	\begin{equation}
	\label{eq:all_adjacent_exp}
	\begin{aligned}
	\mathbb{E}\left[\degv[v][\tree][\omega]\right]&=\operatorname{Tr}\left(\underbrace{\Lv[1]^{[v]}}_{\kappa_1\kappa_2 I^{[v]}}\left(L_{G}^{[v]}\right)^{-1}\right)=\operatorname{Tr}\left(\kappa_1\kappa_2\left(L_{G}^{[v]}\right)^{-1}\right)\\
	&=\sum_{i=1}^{|V|-1}\frac{\kappa_1\kappa_2}{\mu_i}=\sum_{i=1}^{|V|-1}\frac{\kappa_1\kappa_2}{\lambda_i+\kappa_1}.
	\end{aligned}
	\end{equation}
	We have utilized the fact that the trace of a matrix equals the sum of its eigenvalues, and that the eigenvalues of the inverse of a matrix are the reciprocals of the eigenvalues of the original matrix. Analogously, for the variance, we obtain:
	\begin{equation}
		\label{eq:all_adjacent_var}
		\begin{aligned}
			\Var\left[\degv[v][\tree][\omega]\right]
			&=\operatorname{Tr}\left[\kappa_1\kappa_2^2 \left(L_G^{[v]}\right)^{-1}-\kappa_1^2\kappa_2^2 \left(L_G^{[v]}\right)^{-2}\right]\\
			&=\sum_{i=1}^{|V|-1}\frac{\kappa_1\kappa_2^2}{\lambda_i+\kappa_1}-\left(\frac{\kappa_1\kappa_2}{\lambda_i+\kappa_1}\right)^2=\sum_{i=1}^{|V|-1}\frac{\kappa_1\kappa_2^2\lambda_i}{\left(\lambda_i+\kappa_1\right)^2}
		\end{aligned}
	\end{equation}

\end{example}

For a complete unweighted graph with $N$ nodes, i.e. $w(e)=\omega(e)=1 \ \forall e$, the eigenvalues of its Laplacian matrix are $\lambda_1=0$ and $\lambda_i=N$ for $2\leq i \leq N$. As a consequence of equations \eqref{eq:all_adjacent_exp} and \eqref{eq:all_adjacent_var}, we can determine the expectation and variance of a node's degree in a random spanning tree with $N$ nodes:
\begin{align*}
	\mathbb{E}\left[\degv[v][\tree]\right]&=\sum_{i=1}^{N-1}\frac{1}{\lambda_i+1}=1+\sum_{i=1}^{N-2}\frac{1}{N-1+1}=1+\frac{N-2}{N},\\
	\Var\left[\degv[v][\tree]\right]&=\sum_{i=1}^{N-1}\frac{\lambda_i}{\left(\lambda_i+1\right)^2}=\sum_{i=1}^{N-2}\frac{N-1}{\left(N-1+1\right)^2}=\frac{(N-2)(N-1)}{N^2}.
\end{align*}
These values match the expectation and variance of $1+\operatorname{Binomial}(N-2,\frac{1}{N})$, which represents the degree distribution of a fixed vertex in a uniformly random spanning tree with $N$ nodes \cite{pozrikidis2016node}.

	\section{Relation Between Edge Probability and Expected Node Degree in Spanning Trees}\label{sec:relationship-between-edge-probability-and-expected-node-degree-in-spanning-trees}

In this section, we will establish the relationship between the expected weighted degree and the presence probability of the edges incident to node $v$ in a spanning tree. Specifically, we show that the weighted sum of edge presence probabilities, when summed over the edges incident to $v$, equals the expected degree.

\begin{theorem}
	\label{th:Expected_node_degree}
	Let $G=(V,E,w,\omega)$ be an undirected, probability-weighted connected graph. The expected degree of a node $v$ in a random spanning tree $\tree\in \Setspt[G]$ is given by
	\[\mathbb{E}\left[\degv[v][\tree][\omega]\right]=\sum_{k\in\degv[v][\tree][\omega]}k\cdot\Prob\left(\tree\in \SetsptEw[k][G,v]\right)=\sum_{e\in\mathcal{E}_{G}(v)}\omega(e)\Prob(e\in \tree).\]
	Here, $\Prob(e\in \tree)$ represents the probability of edge $e$ being present in a spanning tree of $G$.
\end{theorem}%
\begin{proof}
	First, we note that the random variable $\degv[v][\tree][\omega]$ can be expressed as the weighted sum of the Bernoulli random variables, $X_e$, representing the presence of the edges $e\in\mathcal{E}_{G}(v)$ in a random spanning tree. That is $\degv[v][\tree][\omega]=\sum_{e\in\mathcal{E}_{G}(v)}\omega(e)X_e$. Using the fact that the expectation is a linear operator, 
     we have
     $$\mathbb{E}\left[\degv[v][\tree][\omega]\right]=\mathbb{E}\left[\sum_{e\in\mathcal{E}_{G}(v)}\omega(e)X_e\right]=\omega(e)\sum_{e\in\mathcal{E}_{G}(v)}\mathbb{E}\left[X_e\right]=\sum_{e\in\mathcal{E}_{G}(v)}\omega(e)\Prob(e\in \tree).$$
	\vphantom{.}
\end{proof}

Although expressions for the presence probability of an edge in a random spanning tree are already established \cite{lyonsProbabilityTreesNetworks2016a}, \cref{th:edge_probability} illustrates how this probability can be derived using \cref{th:Expectation_Variance_Laplacian}. 
\begin{corollary}
	\label{th:edge_probability}
	Given a connected, weighted, undirected graph $G=(V,E,w)$, let $\Setspt[G]$ stand for the set of all spanning trees of $G$. For an edge $\hat{e}=(u,v)\in E$, the probability that $\hat{e}$ belongs to a spanning tree $\tree\in\Setspt[G]$ is given by
	\begin{equation}
	\Pr(\hat{e}\in \tree)=\begin{cases}
	w(\hat{e})\left(\ell_{uu}^{-1,[r]}+\ell_{vv}^{-1,[r]}-2\ell_{uv}^{-1,[r]}\right) &\text{ if } r\neq u,v\\
	w(\hat{e})\ell_{uu}^{-1,[v]} &\text{ if } r=v  \\
	w(\hat{e})\ell_{vv}^{-1,[u]} &\text{ if } r=u\\
	\end{cases},
	\label{eq:T_e_weighted_remove_r}
	\end{equation}
	where $\ell_{ij}^{-1,[r]}$ denotes the entry $ij$ of the inverse of the matrix $L_G^{[r]}$ (the Laplacian $L_G$ after removing the row and the column corresponding to node $r$).
\end{corollary}	
\begin{proof}
	By setting the degree weight function equal to 
	$$\omega(e)=\begin{cases}
		1 &\text{ if } e=\hat{e},\\
		0 &\text{otherwise}
	\end{cases},$$
	we find that expected degree of $v$ (or $u$) is equal to the probability that $\hat{e}$ is present in a random spanning tree. Applying \cref{th:Expectation_Variance_Laplacian}, we obtain the desired result. 
\end{proof}

    \section{Extension to Directed Graphs}\label{sec:extension-to-directed-graphs}

We can extend the presented results to the directed case, where the probability and degree weight functions are asymmetric. In other words, we consider graphs where $w(i,j)\neq w(j,i)$ and/or $\omega(i,j)\neq \omega(j,i)$ for any arbitrary vertices $i$ and $j$. In this section, we will briefly expose how the theorems generalize when a directed graph is considered.

The cornerstone of our derivations is the Matrix Tree Theorem for undirected graphs. The Matrix Tree Theorem can be generalized to directed graphs by considering incoming directed trees, in-trees for brevity, rooted at a node $r$ \cite{leenheer2019}. 
\begin{definition}\label{def:in-tree}
	Let $\tree=(V_{\tree},E_{\tree})$ be a directed graph. We say that $\tree$ is an incoming directed spanning tree rooted at $r\in V_{\tree}$ (in-tree for short), if for any vertex $u\in V_{\tree}$ there is exactly one directed path from $u$ to the root $r$, and the root does not have any out-going edges. Note that an in-tree becomes a tree in the classical sense if the directions of the edges are ignored.
	
	We say that an in-tree is spanning on $G=(V,E)$ if $\tree$ is a subgraph of $G$ and $V_{\tree}=V$. The set of spanning in-trees of $G$ rooted at $r$ will be denoted by $\SetTinc[r]$.
\end{definition}
The directed version of the Matrix Tree Theorem \cite{leenheer2019} asserts that
\[w(\SetTinc[r])=\sum_{\tree\in\SetTinc[r]}\prod_{e\in E_{\tree}}w(e)=\det(L_G^{[r]}),\]
where the Laplacian for directed graphs is defined as $L_G\coloneqq D-A^{\top}$ with $A$ being the vertex-adjacency matrix of $G$ represented entry-wise as $A_{ij}=w\big((i,j)\big)$ and $D$ denoting the diagonal matrix defined as $D_{ii}=\sum_{j\in V}A_{ij}$, i.e., $D_{ii}$ is the out-degree of vertex $i$.

Consequently, it can be inferred from the preceding theorems, with appropriate adjustments, that the same formulae can be applied to compute the expected weighted degree of any node $v$ and its variance in a random spanning in-tree rooted at $r$. This holds true when considering the analogous probability distribution defined over the spanning trees, as stated in \eqref{eq:prob_tree_def}, but in this case applied to all in-trees rooted at $r$.

	\section{Conclusion}\label{sec:degree_conclusion}
	We have provided analytical expressions to compute the expected degree of a node in a random spanning tree of a weighted graph, as well as its variance and covariance. These are expressed in terms of the inverse of the Laplacian, once a row and a column have been removed. In addition, we have also showed in \cref{sec:extension-to-directed-graphs} that the same expressions extend to the directed case, where instead spanning incoming trees are considered.

We hope that understanding the expected degree of a node in random spanning trees will inform the future design of methods for analyzing various network dynamics and processes, such as information diffusion, routing efficiency, and resilience to failures.
	
	\section*{Acknowledgments}
This work is supported by the Deutsche Forschungsgemeinschaft (DFG, German Research Foundation) under Germany's Excellence Strategy EXC 2181/1 - 390900948 (the Heidelberg STRUCTURES Excellence Cluster).

	\bibliographystyle{abbrv}
	\bibliography{Degree_node_spanning_tree}

\begin{thebibliography}{10}

\bibitem{albert2002statistical}
R.~Albert and A.-L. Barab{\'a}si.
\newblock Statistical mechanics of complex networks.
\newblock {\em Reviews of modern physics}, 74(1):47, 2002.

\bibitem{barabasiEmergenceScalingRandom1999}
A.-L. Barabási and R.~Albert.
\newblock Emergence of {Scaling} in {Random} {Networks}.
\newblock {\em Science}, 286(5439):509--512, Oct. 1999.

\bibitem{boccaletti2006complex}
S.~Boccaletti, V.~Latora, Y.~Moreno, M.~Chavez, and D.-U. Hwang.
\newblock Complex networks: Structure and dynamics.
\newblock {\em Physics reports}, 424(4-5):175--308, 2006.

\bibitem{bollobasProbabilisticProofAsymptotic1980}
B.~Bollobás.
\newblock A {Probabilistic} {Proof} of an {Asymptotic} {Formula} for the {Number} of {Labelled} {Regular} {Graphs}.
\newblock {\em European Journal of Combinatorics}, 1(4):311--316, Dec. 1980.

\bibitem{chungAverageDistancesRandom2002}
F.~Chung and L.~Lu.
\newblock The average distances in random graphs with given expected degrees.
\newblock {\em Proceedings of the National Academy of Sciences}, 99(25):15879--15882, Dec. 2002.
\newblock Publisher: Proceedings of the National Academy of Sciences.

\bibitem{devriendtGraphsNonnegativeResistance2024}
K.~Devriendt.
\newblock Graphs with nonnegative resistance curvature, Oct. 2024.
\newblock arXiv:2410.07756.

\bibitem{devriendtDiscreteCurvatureGraphs2022}
K.~Devriendt and R.~Lambiotte.
\newblock Discrete curvature on graphs from the effective resistance.
\newblock {\em Journal of Physics: Complexity}, 3(2):025008, June 2022.

\bibitem{Iacobelli2022EpicenterOR}
G.~Iacobelli and D.~R. Figueiredo.
\newblock Epicenter of random epidemic spanning trees on finite graphs.
\newblock {\em Mathematical Modelling of Natural Phenomena}, 2022.

\bibitem{PhysRevResearch.3.023161}
F.~Kaiser and D.~Witthaut.
\newblock Topological theory of resilience and failure spreading in flow networks.
\newblock {\em Phys. Rev. Res.}, 3:023161, Jun 2021.

\bibitem{kirchhoff1847ueber}
G.~Kirchhoff.
\newblock Ueber die aufl{\"o}sung der gleichungen, auf welche man bei der untersuchung der linearen vertheilung galvanischer str{\"o}me gef{\"u}hrt wird.
\newblock {\em Annalen der Physik}, 148(12):497--508, 1847.

\bibitem{laiAnalysisIdentificationMethods2022}
Q.~Lai and H.-H. Zhang.
\newblock Analysis of identification methods of key nodes in transportation network.
\newblock {\em Chinese Physics B}, 31(6):068905, June 2022.

\bibitem{lamSimpleMeasureNetwork2014}
H.~T. Lam and K.~Y. Szeto.
\newblock Simple {Measure} of {Network} {Reliability} {Using} the {Variance} of the {Degree} {Distribution}.
\newblock volume 286, pages 293--302, Cham, 2014. Springer International Publishing.
\newblock Book Title: Proceedings of the Ninth International Conference on Dependability and Complex Systems DepCoS-RELCOMEX. June 30 – July 4, 2014, Brunów, Poland Series Title: Advances in Intelligent Systems and Computing.

\bibitem{leenheer2019}
P.~D. Leenheer.
\newblock An elementary proof of a matrix tree theorem for directed graphs.
\newblock {\em SIAM}, Aug. 2020.
\newblock Publisher: Society for Industrial and Applied Mathematics.

\bibitem{lyonsProbabilityTreesNetworks2016a}
R.~Lyons and Y.~Peres.
\newblock {\em Probability on {Trees} and {Networks}}.
\newblock Cambridge University Press, 1 edition, Oct. 2016.

\bibitem{magnusMatrixDifferentialCalculus2019}
J.~R. Magnus and H.~Neudecker.
\newblock Matrix {Differential} {Calculus} with {Applications} in {Statistics} and {Econometrics}.
\newblock In {\em Matrix {Differential} {Calculus} with {Applications} in {Statistics} and {Econometrics}}, Wiley {Series} in {Probability} and {Statistics}, chapter~8. Wiley, 1 edition, 2019.

\bibitem{moriRandomTrees2005}
T.~Móri.
\newblock On random trees.
\newblock {\em Studia Scientiarum Mathematicarum Hungarica}, 39(1-2):143--155, July 2005.
\newblock Publisher: Akadémiai Kiadó Section: Studia Scientiarum Mathematicarum Hungarica.

\bibitem{newman2003structure}
M.~E. Newman.
\newblock The structure and function of complex networks.
\newblock {\em SIAM review}, 45(2):167--256, 2003.

\bibitem{pozrikidis2016node}
C.~Pozrikidis.
\newblock Node degree distribution in spanning trees.
\newblock {\em Journal of Physics A: Mathematical and Theoretical}, 49(12):125101, 2016.

\bibitem{rapoportCorrelationsUniformSpanning2023}
A.~Rapoport.
\newblock Correlations in uniform spanning trees: a fermionic approach, Dec. 2023.
\newblock arXiv:2312.14992 [math-ph].

\bibitem{rudasRandomTreesGeneral2006}
A.~Rudas, B.~Toth, and B.~Valko.
\newblock Random {Trees} and {General} {Branching} {Processes}, Mar. 2006.
\newblock arXiv:math/0503728.

\bibitem{shinodaUniformSpanningTrees2015}
M.~Shinoda, E.~Teufl, and S.~Wagner.
\newblock Uniform spanning trees on {Sierpinski} graphs, Jan. 2015.
\newblock arXiv:1305.5114 [math].

\bibitem{shiozawaDataSynchronizationNode2022}
K.~Shiozawa, T.~Miyano, and I.~T. Tokuda.
\newblock Data synchronization via node degree in a network of coupled phase oscillators.
\newblock {\em Nonlinear Theory and Its Applications, IEICE}, 13(3):534--543, 2022.

\bibitem{Tutte1984}
W.~T. Tutte.
\newblock {\em Graph theory}.
\newblock Encyclopedia of mathematics and its applications. Addison-Wesley, 1984.

\bibitem{wangStatisticalIdentificationImportant2021}
P.~Wang.
\newblock Statistical {Identification} of {Important} {Nodes} in {Biological} {Systems}.
\newblock {\em Journal of Systems Science and Complexity}, 34(4):1454--1470, 2021.

\bibitem{willemain2002distribution}
T.~R. Willemain and M.~V. Bennett.
\newblock The distribution of node degree in maximum spanning trees.
\newblock {\em Journal of Statistical Computation and Simulation}, 72(2):101--106, 2002.

\bibitem{zhangIdentificationKeyNodes2024}
Y.~Zhang, S.~Zheng, and Y.~Chen.
\newblock Identification of {Key} {Nodes} in {Comprehensive} {Transportation} {Network}: {A} {Case} {Study} in {Beijing}-{Tianjin}-{Hebei} {Urban} {Agglomeration}, {China}.
\newblock {\em Transportation Research Record: Journal of the Transportation Research Board}, 2678(5):827--840, May 2024.

\bibitem{zmigrodEfficientComputationExpectations2021}
R.~Zmigrod, T.~Vieira, and R.~Cotterell.
\newblock Efficient {Computation} of {Expectations} under {Spanning} {Tree} {Distributions}.
\newblock {\em Transactions of the Association for Computational Linguistics}, 9:675--690, Dec. 2021.
\newblock Publisher: MIT Press.

\end{thebibliography}
	
	\appendix
	\section{Covariance}\label{sec:covariance}
In this section, we will expose how to compute the covariance between $ \degv[v][\tree][\omega]$ and $\degv[u][\tree][\omega]$ for two nodes $u$ and $v$. The idea of the proof follows a similar approach to the one exposed in the main part. Namely, we show that the covariance can also be expressed in terms of the derivatives of a polynomial, which is also closely related to the Laplacian matrix of a scaled version of $G$. However, the polynomial differs from the one defined in \cref{def:polynomials_tree_degree}, as it depends on two variables and generalizes the previous polynomial.

Given two nodes $u,v\in V$, let
\begin{equation}\label{eq:poly_cov}
	\polyxy[v,u](x,y)\coloneqq\sum_{\substack{k_1\in \Degv[v]\\k_2\in \Degv[u]}}(-1)^{k_1+k_2}w\left(\SetsptEw[k_1,k_2][G,v,u]\right)x^{k_1}y^{k_2}.
\end{equation}
 Here, $\SetsptEw[k_1,k_2][G,v,u]$ represents the set of spanning trees where node $v$ has degree $k_1$ and node $u$ degree $k_2$, i.e.
$$\SetsptEw[k_1,k_2][G,v,u]\coloneqq\{\tree\in\Setspt[G]\ :\ \degv[v][G][\omega]=k_1, \ \degv[u][G][\omega]=k_2\}.$$

Since $\sum_{\substack{k_2\in \Degv[u]}}w\left(\SetsptEw[k_1,k_2][G,v,u]\right)=w\left(\SetsptEw[k_1][G,v]\right)$, we can see indeed that the polynomial $\polyxy[v,u]$ generalizes the polynomials $\polyw[v]$ and $\polyw[u]$ because
\begin{equation}\label{eq:relation_polyxy_polyw}
	\begin{aligned}
	\polyxy[v,u](x,-1)&=\polyw[v](x)\\
	\polyxy[v,u](-1,y)&=\polyw[u](y)
	\end{aligned}
\end{equation}
Moreover, similarly to \cref{lem:polynomials_count_trees} and equation \eqref{eq:poly_Laplacian_relation1}, we can also apply the Matrix Tree Theorem to relate the polynomial with the determinant of a Laplacian of a scaled version of $G$. Namely 
\begin{equation}\label{eq:polyxy_laplacian}
	\det(L_{\polyxyG})=\wpolyxyG\left(\Setspt[\polyxyG]\right)=\polyxy[u,v](-\alpha,-\beta)
\end{equation}
where $\polyxyG\coloneqq  (G\times_v \alpha^{\omega})\times_u\beta^{\omega}$, whose associated weight probability function is represented by $\wpolyxyG$. 

From here, we can apply a similar approach to the one shown for the expectation and variance in \cref{th:Expectation_Variance_Laplacian} and relate the derivatives of the Laplacian $L_{\polyxyG}$ to the covariance between $\degv[v][\tree][\omega]$ and $\degv[u][\tree][\omega]$.

Indeed, we have that 
\begin{equation}\label{eq:cov_poly}
	\begin{aligned}
		\frac{\partial \log(\polyxy[u,v](-x,-y))}{\partial x \partial y}\Biggr|_{(x,y)=(1,1)}&=\left[\frac{\frac{\partial \polyxy[u,v](-x,-y)}{\partial x \partial y}}{\polyxy[u,v](-x,-y)}-\frac{\frac{\partial \polyxy[u,v](-x,-y)}{\partial x }}{\polyxy[u,v](-x,-y)}\frac{\frac{\partial \polyxy[u,v](-x,-y)}{\partial y }}{\polyxy[u,v](-x,-y)}\right]\Biggr|_{(1,1)}\\
		&\overset{*}{=}\frac{\frac{\partial \polyxy[u,v](-x,-y)}{\partial x \partial y}}{\polyxy[u,v](-x,-y)}\Biggr|_{(x,y)=(1,1)}-\frac{\polyw[v]'(-1)}{\polyw[v](-1)}\frac{\polyw[u]'(-1)}{\polyw[u](-1)}\\
		&\overset{\star}{=}\mathbb{E}\left[\degv[v][\tree][\omega]\degv[u][\tree][\omega]\right]-\mathbb{E}\left[\degv[v][\tree][\omega]\right]\mathbb{E}\left[\degv[u][\tree][\omega]\right]\\&=\Cov(\degv[v][\tree][\omega],\degv[u][\tree][\omega]).
	\end{aligned}	
\end{equation}
In $*$, we have used the equalities stated in \eqref{eq:relation_polyxy_polyw}. For the second summand of the equality marked with $\star$, we applied \cref{th:poly_expvar}. The first summand equality comes from
\begin{equation}
	\begin{aligned}
		\frac{\partial\polyxy[u,v](-x,-y)}{\partial x \partial y}&=\sum_{\substack{k_1\in \Degv[v]\\k_2\in \Degv[u]}}(-1)^{k_1+k_2+2}k_1k_2w\left(\SetsptEw[k_1,k_2][G,v,u]\right)x^{k_1-1}y^{k_2-1} \xrightarrow[y=1]{x=1}\\
		\frac{\partial(\polyxy[u,v](-1,-1))}{\partial x \partial y}&=\sum_{\substack{k_1\in \Degv[v]\\k_2\in \Degv[u]}}k_1k_2w\left(\SetsptEw[k_1,k_2][G,v,u]\right)\\
  &=\omega(\Setspt)\sum_{\substack{k_1\in \Degv[v]\\k_2\in \Degv[u]}}k_1k_2\frac{w\left(\SetsptEw[k_1,k_2][G,v,u]\right)}{\omega(\Setspt)}\\
  &=\omega(\Setspt)\mathbb{E}\left[\degv[v][\tree][\omega]\degv[u][\tree][\omega]\right]=\polyxy[u,v](-1,-1)\mathbb{E}\left[\degv[v][\tree][\omega]\degv[u][\tree][\omega]\right]
	\end{aligned}
\end{equation}

Finally, by combining equations \eqref{eq:polyxy_laplacian} and \eqref{eq:cov_poly}, we obtain the result after computing the corresponding derivatives, a step we will omit:
\begin{equation}\label{eq:covariance_laplacian}
	\begin{aligned}
		\Cov(\degv[v][\tree][\omega],\degv[u][\tree][\omega])&=\frac{\partial \log(\polyxy[u,v](-\alpha,-\beta))}{\partial \alpha \partial \beta}\Biggr|_{(\alpha,\beta)=(1,1)}=\frac{\det(L_{\polyxyG})}{\partial \alpha \partial \beta}\Biggr|_{(\alpha,\beta)=(1,1)}\\
		&=\operatorname{Tr}\left[\left(\left(\Lvu[2][v][u]\right)^{[r]}-\left(\Lv[1]\right)^{[r]} \left(L_G^{[r]}\right)^{-1}\left(\Lu[1]\right)^{[r]} \right)\left(L_G^{[r]}\right)^{-1}\right]
	\end{aligned}
\end{equation}
Here $\left(\Gwp[2]_{v}\right)_{u}\coloneqq \left(G_v\times_v \omega^2\right)_u$ is the graph containing the edge $(u,v)$, if it exists in $G$, otherwise is the empty graph with no edges. 

\begin{remark}
	Note that given $e_1=(v_1,u_1),e_2=(v_2,u_2)\in E$ and the degree weight function 
	$$\omega(e)=\begin{cases}1 &e\in\{e_1,e_2\}\\ 0 &\text{otherwise}\end{cases},$$
	the covariance between \( \degv[v_1][\tree][\omega] \) and \( \degv[v_2][\tree][\omega] \) equals the covariance of the presence of the two edges in a random spanning tree.
\end{remark}

\section{Expectations of Spanning Trees over Additively Decomposable Functions}\label{sec:add_decomposable}
Given a weighted graph $G=(V,E,w,\omega)$, the work of Zmigrod et al. \cite{zmigrodEfficientComputationExpectations2021} proposes a unified approach to computing expectations over spanning trees when the functions are additively decomposable along the edges of the tree. Formally, they provide an algorithm to compute the $\mathbb{E}[r(\tree)]$, where $T\in\Setspt[G]$ is sampled with probability proportional to its weight, as described in \cref{eq:prob_tree_def}, and where $r$ is additively decomposable along the edges,  meaning it can be expressed as $r(\tree)=\sum_{e\in E_{\tree}}\omega(e)$. Under such functions, one can compute values such as the Shannon entropy of the spanning tree probability distribution or the KL-divergence between two spanning trees distributions.

Thanks to the linearity of the expectation operator, $\mathbb{E}[r(\tree)]$ can be expressed as a sum of expectations over the node degrees.
\begin{equation}
    \label{eq:expec_decomp=sum_expec_deg}
    \sum_{v\in V}\mathbb{E}\left[\degv[v][\tree][\omega]\right]= \sum_{v\in V}\mathbb{E}\left[\sum_{e\in\mathcal{E}_{G}(v)}\omega(e)X_e\right]=\mathbb{E}\left[2\sum_{e\in E_T}\omega(e)X_e(T)\right]=2\mathbb{E}\left[\sum_{e\in E_T}r(\tree)\right]
\end{equation}
where $X_e(\tree)$ are Bernoulli random variables representing the presence of the edges $e$ in a random spanning tree. 
Thus, our framework, can also be used to express expectations over additively decomposable functions, as the ones proposed in \cite{zmigrodEfficientComputationExpectations2021}. However, our approach allows these formulas to be expressed in terms of the Laplacian in a more compact manner, as shown in \cref{th:lap_expec_add_decomp.}. Moreover, we are able to express the variance of additively decomposable functions in closed form, which was not addressed in \cite{zmigrodEfficientComputationExpectations2021}.

\begin{corollary}\label{th:lap_expec_add_decomp.}
    Given a weighted graph $G=(V,E,w,\omega)$ and the additively decomposable function along the edges of the tree associated with the degree-weight function $\omega$, i.e. $r(\tree)=\sum_{e\in E_{\tree}}\omega(e)$, then
    \begin{align}
        \mathbb{E}\left[r(\tree)\right]&=\operatorname{Tr}\left[L_{G\times\omega}^{[r]}\left(L_{G}^{[r]}\right)^{-1}\right],\label{eq:exp_Lap_add_decomp}\\
        \Var\left[r(\tree)\right]&=\operatorname{Tr}\left[\left(L_{G\times\omega^2}^{[r]}-L^{[r]}_{G\times\omega}\left(L_{G}^{[r]}\right)^{-1}L^{[r]}_{G\times\omega}\right)\left(L_{G}^{[r]}\right)^{-1}\right],\label{eq:var_Lap_add_decomp}
    \end{align}
    where $L_{G}$, $L_{G\times\omega}$ and $L_{G\times\omega^2}$ denote the Laplacian of $G$ using the edge weights given by $w$, the product of $\omega$ and $w$ and the product of $\omega^2$ and $w$, respectively.
\end{corollary}
\begin{proof}
    \Cref{eq:exp_Lap_add_decomp} is a direct consequence of applying \cref{th:Expectation_Variance_Laplacian} to \cref{eq:expec_decomp=sum_expec_deg}:
    \begin{equation*}
        \begin{aligned}
            \mathbb{E}\left[r(\tree)\right]&\underbrace{=}_{\text{\cref{{eq:expec_decomp=sum_expec_deg}}}}\frac{1}{2}\sum_{v\in V}\mathbb{E}\left[\degv[v][\tree][\omega]\right]\underbrace{=}_{\text{\cref{{th:Expectation_Variance_Laplacian}}}}
    \frac{1}{2}\sum_{v\in V}\operatorname{Tr}\left[\left(\Lv[1]\right)^{[r]}\left(L_{G}^{[r]}\right)^{-1}\right]\\
    &=\frac{1}{2}\operatorname{Tr}\left[\sum_{v\in V}\left(\Lv[1]\right)^{[r]}\left(L_{G}^{[r]}\right)^{-1}\right]=
    \frac{1}{2}\operatorname{Tr}\left[2L_{G\times\omega}^{[r]}\left(L_{G}^{[r]}\right)^{-1}\right]\\
    &=\operatorname{Tr}\left[L_{G\times\omega}^{[r]}\left(L_{G}^{[r]}\right)^{-1}\right]
        \end{aligned}
    \end{equation*}
    Analogously, \cref{eq:exp_Lap_add_decomp} follows from  \cref{th:Expectation_Variance_Laplacian} combined with the variance sum law
    \begin{equation*}
        \begin{aligned}
        \Var\left[r(\tree)\right]&=\Var\left[\frac{1}{2}\sum_{v\in V}\degv[v][\tree][\omega]\right]=\frac{1}{4}\sum_{v\in V} \Var\left[\degv[v][\tree][\omega]\right]+\frac{1}{4}\sum_{\substack{v,u\in V\\v\neq u}}\Cov\left[\degv[v][\tree][\omega],\degv[u][\tree][\omega]\right]\\
        &=\frac{1}{4}\sum_{v\in V}\operatorname{Tr}\left[\left(\left(\Lv[2]\right)^{[r]}-\left(\Lv[1]\right)^{[r]} \left(L_G^{[r]}\right)^{-1}\left(\Lv[1]\right)^{[r]} \right)\left(L_G^{[r]}\right)^{-1}\right]\\
        &+\frac{1}{4}\sum_{\substack{v,u\in V\\v\neq u}}\operatorname{Tr}\left[\left(\left(\Lvu[2][v][u]\right)^{[r]}-\left(\Lv[1]\right)^{[r]} \left(L_G^{[r]}\right)^{-1}\left(\Lu[1]\right)^{[r]} \right)\left(L_G^{[r]}\right)^{-1}\right]\\
        &=\operatorname{Tr}\left[\left(L_{G\times\omega^2}^{[r]}-L^{[r]}_{G\times\omega}\left(L_{G}^{[r]}\right)^{-1}L^{[r]}_{G\times\omega}\right)\left(L_{G}^{[r]}\right)^{-1}\right]
        \end{aligned}
    \end{equation*}    
    \phantom{-}
\end{proof}

\end{document}